\documentclass[smallextended]{svjour3}       % onecolumn (second format)
\smartqed  % flush right qed marks, e.g. at end of proof
\usepackage[left=1in,top=1in,right=1in,bottom=1in,letterpaper]{geometry}

      \usepackage[backref,colorlinks=true, linkcolor=blue, citecolor=blue, urlcolor=blue, pdfborder={0 0 0}]{hyperref}

      \hypersetup{pdftitle={One condition for solution uniqueness and robustness of both l1-synthesis and l1-analysis minimizations}, pdfauthor={Hui Zhang, Ming Yan, Wotao Yin}}
    \usepackage[pdftex]{graphicx}
\usepackage[usenames]{color}

\usepackage[color,final]{showkeys}

\usepackage{amsmath,amssymb}
\usepackage{latexsym,amsfonts,amscd,amsxtra,amstext}
\usepackage{algorithm,algorithmic}
\usepackage{url}
\usepackage{hyperref}

\usepackage[normalem]{ulem} % to use \sout

\newcommand{\cF}{{\mathcal{F}}}

\newcommand{\cT}{{\mathcal{T}}}

\newcommand{\RR}{\mathbb{R}}

\newcommand{\sign}{\mathrm{sign}}

\newcommand{\supp}{{\mathrm{supp}}} % support
\newcommand{\bc}{\begin{center}}
\newcommand{\ec}{\end{center}}

\newcommand{\bdm}{\begin{displaymath}}
\newcommand{\edm}{\end{displaymath}}

\newcommand{\beq}{\begin{equation}}
\newcommand{\eeq}{\end{equation}}

\newcommand{\bfl}{\begin{flushleft}}
\newcommand{\efl}{\end{flushleft}}

\newcommand{\bt}{\begin{tabbing}}
\newcommand{\et}{\end{tabbing}}

\newcommand{\beqn}{\begin{eqnarray}}
\newcommand{\eeqn}{\end{eqnarray}}

\newcommand{\beqs}{\begin{align*}} % no equation numbers
\newcommand{\eeqs}{\end{align*}}  % no equation numbers

\DeclareMathOperator*{\Min}{minimize}
\DeclareMathOperator*{\Max}{maximize}
\newcommand{\st}{\mbox{subject to}}
\newtheorem{condition}{Condition}
\newtheorem{assumption}{Assumption}

\begin{document}

\title{One condition for solution uniqueness and robustness of both l1-synthesis and l1-analysis minimizations \thanks{Communicated by xxx}}

\titlerunning{Solution uniqueness/robustness of $\ell_1$-analysis minimization}        % if too long for running head

\author{Hui Zhang \and Ming Yan \and Wotao Yin}
%\journalname{}


\institute{Hui Zhang \at
              Department of Mathematics and Systems Science, College of Science,\\
              (also affiliated with the state key laboratory for high performance computation,)\\
              National University of Defense Technology,\\
              Changsha, Hunan, China, 410073.\\
              \email{hhuuii.zhang@gmail.com}
           \and
              Ming Yan and Wotao Yin \at
              Department of Mathematics,\\
              University of California,\\
              Los Angeles, CA 90095, USA.\\
              \email{yanm@math.ucla.edu (M. Yan) and wotaoyin@math.ucla.edu (W. Yin)}            % \\
}

\date{Received:  / Accepted: }

\maketitle

\begin{abstract}
The $\ell_1$-synthesis model and the $\ell_1$-analysis model  recover structured signals from their undersampled measurements. The solution of  former  is  a sparse sum of dictionary atoms, and that of the latter   makes sparse correlations with  dictionary atoms. This paper addresses the question: when  can we trust these models to recover specific signals? We answer the question with a condition that is both necessary and sufficient to guarantee the recovery to be unique and exact and,  in presence of measurement noise,  to be robust. The condition is one--for--all in the sense that it applies to both of the $\ell_1$-synthesis and $\ell_1$-analysis  models,   to both of their constrained and unconstrained formulations, and to both the exact recovery and robust recovery cases. Furthermore,  a convex infinity--norm program is introduced for numerically verifying the condition. A comprehensive comparison with related existing conditions are included.
\end{abstract}

\section{Introduction}
With both synthesis and analysis models of sparse signal recovery \cite{EMR}, one is interested in when the recovery is successful, namely, whether the solution is unique and whether the solution error is proportional to the amount of noise in the measurements. 

Various \emph{sufficient} conditions have been proposed to guarantee successful recovery. Let us first discuss two types of recovery conditions. The \emph{uniform recovery} conditions guarantee not only the successful recovery of one sparse signal but all  signals that are sufficiently sparse, irrelevant of  the locations of their nonzero entries. The \emph{non-uniform recovery conditions}, however, focus on the recovery  of a restricted set of sufficiently sparse signals, for example, the sparse signals with a specific support. Since non-uniform recovery conditions cover fewer  signals, they are in general weaker and thus easier to hold  than  uniform recovery conditions. In addition, some non-uniform recovery conditions, especially those for fixed signal support, are easier to numerically verify, whereas all the existing uniform recovery conditions are  numerically intractable to verify on given sensing matrices. On the other hand, several kinds of random matrices (such as  those with i.i.d. subgaussian entries) satisfy the uniform recovery conditions with high probability. The uniform conditions are  useful in designing randomized linear  measurements, whereas non-uniform conditions are useful on deterministic linear  measurements on restricted sets of signals. 

Non-uniform conditions for $\ell_1$-synthesis minimization include the non-uniform dual certificate condition \cite{Fuch1} and the ``RIPless" property \cite{CP11}. Well-known examples of uniform conditions include the restricted isometry principle \cite{CT}, the null space condition \cite{Co}, the spherical section property \cite{Zh}, and others.
 
Because $\ell_1$-analysis minimization takes a more general form than $\ell_1$-synthesis minimization, some of the above non-uniform recovery conditions have been extended to the analysis case; recent works~\cite{Gras2,VPDF,Halt,NW1,NW2,KW} have  made significant contributions.

This paper studies the so-called \emph{dual certificate}  condition, a type of non-uniform condition. We show that if a  signal is a solution to any model among \eqref{eq:bp}--\eqref{eq:qpb} described below,  there exists a  necessary and sufficient condition,  same for any of the three models, that guarantees that the signal will be uniquely recovered. While this result  has been partially known in previous work for $\ell_1$ minimization as a sufficient condition, we establish three new results:
\begin{itemize}
\item the condition previously known to be sufficient is in fact necessity;
\item the condition guarantees robustness to noise. That is, under this condition,  if  the observed data is contaminated by arbitrary noise,  the solution to either \eqref{eq:qpa} or \eqref{eq:qpb} is robust to the noise in the sense the solution error is proportional to the Euclidean norm of the noise;
\item  nearly the same condition but imposed on the support of the  largest $|I|$ entries of an approximately sparse signal guarantees its robust recovery by model ~\eqref{eq:qpb}. 
\end{itemize}
These results complete the theory of non-uniform recovery of the $\ell_1$-synthesis and $\ell_1$-analysis models. 

The proposed condition is compared to existing non-uniform conditions in the literature, most of which are stronger than ours and are thus sufficient yet not necessary. Note that some of those stronger ones also give additional properties that ours does not. Technically, a part of our analysis is inspired by  existing results in \cite{BO,Gras2,Gras1,ZYC,rauhut2014,foucart2013}, which bear certain similarity among themselves and will be mentioned in later sections.

The rest of the paper is organized as follows. Sections~\ref{sec:02} and \ref{sec:2} formulate the problem and state the main results. Section~\ref{sec:3} reviews several related results. Section~\ref{sec:4} discusses  condition verification. Proofs for the main results are given in sections~\ref{sec:5}, ~\ref{sec:6}, and~\ref{sec:8}.

\section{Problem Formulation and Contributions}
\label{sec:02}

\subsection{Notation}
We equip $\mathbb{R}^n$ with the canonical scalar product $\langle \cdot, \cdot \rangle$ and Euclidean norm $\|\cdot\|_2$. We let $|\cdot|$ return the cardinality if the input is a set or the absolute value if the input is a number. For any $x\in \mathbb{R}^n$,  $\textrm{supp}(x)=\{k: 1\leq k\leq n,x_k\neq0\}$ is the index set of the non-zero entries of $x$. $\textrm{sign}(x)$ is the vector whose $i$th entry is the sign of $x_i$, taking a value among $+1$, $-1$, and $0$. For any $p\geq 1$, the $\ell_p$-norm of $x\in \mathbb{R}^n$ is
$$\|x\|_p=\left(\sum_{i=1}^n|x_i|^p\right)^{1/p},$$
its $\ell_0$-``norm'' is $\|x\|_0=|\textrm{supp}(x)|$, and its $\ell_\infty$-norm is $\|x\|_\infty=\max\{|x_i|: i=1,\cdots,n\}$. For $x\in \mathbb{R}^n$ and  $I\subset \{1,2,\cdots, n\}$, $x_I$ denotes the vector formed by the entries $x_i$ of $x$ for $i\in I$, and $I^c$ is the complement of $I$. Similarly,  $A_I$ is the submatrix formed by the columns of $A$ indexed by $I$.  $A^T$ is the transpose of $A$. We use $A^T_I$ for the transpose of submatrix $A_I$, not a submatrix of $A^T$. For square matrix $A$,  $\lambda_{\max}(A)$ and $\lambda_{\min}(A)$ denote its largest and smallest eigenvalues, respectively, $\textrm{Cond}(A)$ denotes its condition number, and $\|A\|$ denotes its spectral norm. The null and column spaces of  $A$ are denoted by $\textrm{Ker} (A)$ and $\textrm{Im} (A)$, respectively.

\subsection{Problem formulations}
Let $x^*\in\RR^n$ be a  signal of interest. This paper  studies when  $\ell_1$ minimization can uniquely and robustly  recover $x^*$ from its linear measurements
\begin{equation*}
b=\Phi x^*+w,
\end{equation*}
where $\Phi\in\RR^{m\times n}$ is a certain matrix and  $w\in\RR^m$ is noise. We focus on the  setting $m\leq n$.

The results of this paper cover the following  $\ell_1$ minimization formulations:
\begin{subequations} \label{mod1}
\begin{align}
\label{eq:bp}
& \Min_x\|\Psi^Tx\|_1,\quad\st~\Phi x=b,\\
\label{eq:qpa}
&\Min_x \|\Phi x-b\|_2^2+\lambda \|\Psi^Tx\|_1,\\
\label{eq:qpb}
&\Min_x \|\Psi^Tx\|_1,\quad\st~~\|\Phi x-b\|_2\leq \delta,
\end{align}
\end{subequations}
where $\delta, \lambda$ are positive parameters. Model \eqref{eq:bp} is used if no noise is present, i.e., $w=0$. If $\Psi=Id$, the identify matrix, models in \eqref{mod1} are referred  to as the $\ell_1$ (or more generally, $\ell_1$-synthesis) models. If $\Psi\not=Id$, they are referred to as the $\ell_1$-analysis models, which are recently reviewed in \cite{EMR}.

In  $\ell_1$-synthesis models, the signal of interest is synthesized as $x^* = Dc$, where $D$ is a certain dictionary and $c$ is the sparse coefficients. The $\ell_1$-analysis model, including the cosparse analysis model \cite{NDEG} and the total variation model \cite{ROF} as widely known examples, has recently attracted a lot of attention.  The underlying signal is expected to make sparse correlations with  the columns (atoms) in an (overcomplete) dictionary $\Psi$, i.e., $\Psi^Tx^*$ is sparse; see
\cite{CENR,NDEG,LML}.

\subsection{Geometry}
It is worth noting that the underlying geometry  of solution uniqueness is rather  clear. Indeed, it was characterized in terms of polytope faces by Donoho in \cite{donoho2005}, and also of null spaces and tangent cones by Chandrasekaran et al. in \cite{chand}. However, there lacks a concrete \emph{if-and-only-if} condition for one to check and the robustness results are difficult to obtain from  geometry.

Let us describe the geometry of uniqueness recovery under the simple setting: $\Psi=I$. Then, $x^*$ is the unique solution to \eqref{eq:bp} if the affine set $\cF=\{x:\Phi x=b\}$ touches the $\ell_1$ polytope $\{x:\|x\|_1=\gamma\}$, for some $\gamma>0$, at a \emph{unique} point on a face of the polytope. The uniqueness for \eqref{eq:qpa} and \eqref{eq:qpb} is similar except $\cF$ is replaced by a ``tube'' $\cT=\{x:\|\Phi x-b\|_2\le\sigma\}$. In both cases, the uniqueness is certified by  (i) a kernel condition that ensures touching  the face at just one point, and (ii) the existence of a hyperplane (dual certificate) that separates $\cF$ or $\cT$ from the polytope that creates a positive gap between them outside the face. The same geometry applies to $\ell_1$-analysis minimization but involves the more complicated polytope $\{x:\|\Psi^Tx\|_1=\gamma\}$; see \cite{rauhut2013,foucart2014}.

This work not only formalizes these geometrical explanations with a concrete \emph{if-and-only-if} condition, but also establishes robustness bounds under the same condition, which cannot be explain by the above geometry. Suppose that noise is present in the measurements; then,  the point where  $\cF$ or $\cT$  touches the polytope will move, and the point may or may not stay in the same face in the noise-free case. The above geometry cannot clearly explain how far the point will move due to noise.

\section{Main Condition and Results}
\label{sec:2}

\subsection{Main Condition}
\begin{condition}\label{cond1}
Given $\bar{x}\in \RR^n$, index sets $I=\mathrm{supp}(\Psi^T\bar{x})\subset \{1,\cdots, l\}$ and $J = I^c$ satisfy

(1)~$\mathrm{Ker}(\Psi^T_J) \bigcap \mathrm{Ker}(\Phi)=\{0\}$;

(2)~There exists $y\in \RR^l$ such that $\Psi y\in \mathrm{Im}(\Phi ^T)$,  $y_{I}=\mathrm{sign}(\Psi_I^T\bar{x})$, and $\|y_J\|_\infty<1.$
\end{condition}
It is not difficult to understand the condition as we shall explain below.

Condition \ref{cond1} part (1)  says that there does \emph{not} exist any nonzero $\Delta x$ satisfying both $\Psi^T_J \bar{x} = \Psi^T_J (\bar{x}+\Delta x)$ and $\Phi \bar{x} = \Phi(\bar{x}+\Delta x)$. Otherwise, there exists a nonempty interval $\mathcal{I}=[\bar{x}-\alpha \Delta x, \bar{x}+ \alpha \Delta x]$ for some sufficiently small  $\alpha>0$  so that  $\Phi x=\Phi \bar{x}$ and $\|\Psi^T x\|_1$ is constant for $x\in\mathcal{I}$; hence $\bar{x}$ cannot be the unique minimizer. This condition  implies that the role of $\ell_1$ minimization is ``limited'' to recovering the support set $I$; if an oracle gives us $I$, we must solely rely on $\Psi$, $\Phi$, and $b$ recover $\bar{x}$, since $\|\Psi^T x\|_1$ is locally linear near $\bar{x}$ over $\{x:\mathrm{supp}(\Psi^T{x})=I\}$ and thus lacks the ability to pick $\bar{x}$ out. In addition, part (1) of the condition also implies the existence of $\rho$ and $\tau$ such that
\begin{align}
\|\Psi_I^Tx\|_2\leq \rho \|\Psi_J^Tx\|_1+\tau\|\Phi x\|_2
\end{align}
for all $x\in\RR^n$, which will be used later in Theorem~\ref{thm3}.

Condition \ref{cond1} part (2) states the existence of a strictly-complementary \emph{dual certificate} $y$. To see this, let us check a part of the optimality conditions of (\ref{eq:bp}): $0\in\Psi \partial \|\cdot\|_1(\Psi^Tx)-\Phi^T\beta$, where vector $\beta$ is the Lagrangian multipliers; we can rewrite the condition as $0=\Psi y-\Phi^T\beta$ where $y\in \partial \|\cdot\|_1(\Psi^Tx)$, which translates to $y_{I}=\mathrm{sign}(\Psi_I^Tx)$ and $\|y_J\|_\infty\leq 1$. This $y$ certifies the optimality of $\bar{x}$. For solution uniqueness and/or robustness, we shall later show that the strict inequality $\|y_J\|_\infty< 1$ is necessary.

It is worth mentioning a variant of Condition \ref{cond1}  as follows.
\begin{condition}\label{cond1'} Given $\bar{x}\in\RR^n$, let $I=\mathrm{supp}(\Psi^T\bar{x})$. There exists a nonempty index set $J\subseteq I^c$ such that the index sets $I$, $J$ and $K={(I\bigcup J)^c}$ satisfy

(1)~$\mathrm{Ker}(\Psi^T_J) \bigcap \mathrm{Ker}(\Phi)=\{0\}$;

(2)~There exists $y\in \RR^l$  such that $\Psi y\in \mathrm{Im}(\Phi ^T)$, $y_{I}=\mathrm{sign}(\Psi_I^T\bar{x})$, $\|y_J\|_\infty<1$, and $\|y_K\|_\infty\leq 1$.
\end{condition}
In Condition \ref{cond1'}, a smaller $J$ relaxes part (2) but gives a larger $\mathrm{Ker}(\Psi^T_J)$ and thus tightens part (1). Although Condition \ref{cond1'} allows a more flexible $J$ than Condition \ref{cond1}, we shall show that they are equivalent.

The comparisons between Condition 
\ref{cond1} and those in the existing literature are given in Section \ref{sec:3} below.
\subsection{Main Results}
Depending on the specific models in \eqref{mod1}, we need the following assumptions:
\begin{assumption}\label{assumption1}
Matrix $\Phi$ has full row-rank.
\end{assumption}
\begin{assumption}\label{assumption2}
$\lambda_{\textrm{max}}(\Psi\Psi^T)=1.$
\end{assumption}
\begin{assumption}\label{assumption3}
Matrix $\Psi$ has full row-rank.
\end{assumption}

Assumptions \ref{assumption1} and \ref{assumption3} are standard to avoid redundancy. Assumption \ref{assumption2} is non-essential as we can  scale any general $\Psi$ by multiplying  ${1\over \sqrt{\lambda_{\textrm{max}}(\Psi\Psi^T)}}$.
Below we state our main results and delay their proofs to Sections~\ref{sec:5}--\ref{sec:8}.

\begin{theorem}[Uniqueness]\label{thm1}Under Assumption~\ref{assumption1}, let $\hat{x}$ be a solution to problem (\ref{eq:bp}) or (\ref{eq:qpa}), or under Assumptions~\ref{assumption1} and \ref{assumption3}, let  $\hat{x}$ be a solution to problem (\ref{eq:qpb}). The following are equivalent:

1) Solution $\hat{x}$ is unique;

2) Condition~\ref{cond1} holds for $\bar{x}=\hat{x}$;

3) Condition~\ref{cond1'} holds for $\bar{x}=\hat{x}$.
\end{theorem}

This theorem states that Conditions ~\ref{cond1} and \ref{cond1'} are equivalent, and they are necessary and sufficient for a solution $\hat{x}$ to problem (\ref{eq:bp}), or to problem (\ref{eq:qpa}), or to problem (\ref{eq:qpb}) to be unique.
To state our next result on robustness, we let
$$r(J):=\sup_{u\in \mathrm{Ker}(\Psi^T_J)\backslash \{0\}}\frac{\|u\|_2}{\|\Phi u\|_2}.$$
Part (1) of Condition \ref{cond1} ensures that $0<r(J)<+\infty$. If $\Psi = I$, then $u\in \mathrm{Ker}(\Psi^T_J)\backslash \{0\}$ is a sparse nonzero vector with maximal support $J^c$, so $r(J)$ is the inverse of the minimal singular value of the submatrix $\Phi_{J^c}$.
Below we  claim Condition \ref{cond1} ensures the robustness of problems~(\ref{eq:qpa}) and~(\ref{eq:qpb}) to arbitrary noise in $b$.

\begin{theorem}[Robustness to measurement noise]\label{thm2}
Under Assumptions~\ref{assumption1}-\ref{assumption3}, given an original signal $x^*\in \RR^n$, let $I=\mathrm{supp}(\Psi^Tx^*)$ and $J=I^c$. For arbitrary noise $w$, let $b=\Phi x^*+w$ and  $\delta=\|w\|_2$. If Condition \ref{cond1} is met for $\bar{x}=x^*$, then

1) For any $C_0>0$, there exists constant $C_1>0$ such that every minimizer $x_{\delta,\lambda}$ of problem~(\ref{eq:qpa}) using parameter $\lambda=C_0\delta$ satisfies $$\|\Psi^T(x_{\delta,\lambda}- x^*)\|_1\leq C_1\delta;$$

2) Every minimizer $x_{\delta}$ of problem~(\ref{eq:qpb}) satisfies $$\|\Psi^T(x_{\delta}- x^*)\|_1\leq C_2\delta.$$

The constraints $c_1$ and $C_2$ are given as follows. Define

$$\beta=(\Phi \Phi^T)^{-1}\Phi\Psi y,\quad
C_3=r(J)\sqrt{|I|}\quad\mbox{and}\quad C_4=\frac{1+\mathrm{Cond}(\Psi)\|\Phi\|C_3}{1-\|y_J\|_\infty},$$
with which we have
\begin{align*}
C_1&=2C_3+C_0\|\beta\|_2+\frac{(1+C_0\|\beta\|_2/2)^2C_4}{C_0}, \\
C_2&=2C_3+2C_4\|\beta\|_2.
\end{align*}
\end{theorem}
\begin{remark}
From the results of Theorem \ref{thm2}, it is straightforward to derive  $\ell_1$ or $\ell_2$ bounds for $(x_{\delta,\lambda}- x^*)$ and $(x_{\delta}- x^*)$ under Assumption \ref{assumption2}.
\end{remark}
\begin{remark}
Since $C_0$ is free to choose, one can choose the optimal $C_0 = \sqrt{\frac{4C_4}{4\|\beta\|_2 +C_4 \|\beta\|_2^2}}$ and simplify $C_1$ to
$$C_1 = 2C_3 + C_4 \|\beta\|_2+\sqrt{C_4^2 \|\beta\|_2^2+4 C_4\|\beta\|_2}\leq 2C_3+2C_4\|\beta\|_2+2,$$
which becomes very similar to $C_2$. This reflects the equivalence between problems \eqref{eq:qpa} and \eqref{eq:qpb} in the sense that given $\lambda$, one can find $\delta$ so that they have the same solution, and vice versa.
\end{remark}
\begin{remark}
Both $C_1$ and $C_2$ are the sum of $2C_3$ and other terms. $2C_3$ alone bounds the error when $\Psi^T x_{\delta,\lambda}$ (or $\Psi^T x_{\delta}$) and $\Psi^T x^*$ have matching signs. Since $C_3$ does not depend on $y_J$, part (2) of Condition \ref{cond1} does not play any role, whereas part (1) plays the major role. When the signs of $\Psi^T x_{\delta,\lambda}$ (or $\Psi^T x_{\delta}$) and $\Psi^T x^*$  do \emph{not} match, the remaining terms in $C_1$ and $C_2$ are involved, and they are affected part (2) of Condition \ref{cond1}; in particular, $\|y_J\|_\infty <1$ plays a big role as $C_4$ is inversely proportional to $1-\|y\|_\infty$. Also, since there is no knowledge about the support of $\Psi^T x_{\delta,\lambda}$, which may or may not equal to that of $\Psi^T x^*$, $C_4$ inevitably depends the global properties of $\Psi$ and $\Phi$. In contrast, $C_3$ only depends on the restricted property of $\Phi$.
\end{remark}

When $\Psi^Tx^*$ is only approximately sparse, we need to choose $I$ as the support of the $|I|$ largest entries of $\Psi^Tx^*$ in Condition~\ref{cond1}. Then we have the following robustness result of approximately sparse signals for \eqref{eq:qpb}.
\begin{theorem}[Robustness for approximately sparse signals]\label{thm3}
Under Assumption \ref{assumption3}, given an original signal $x^*\in \RR^n$ with $|I|$ largest entries of $\Psi^Tx^*$ supported on $I$, let $J=I^c$. For arbitrary noise $w$, let $b=\Phi x^*+w$ and  $\delta=\|w\|_2$. If Condition \ref{cond1} is met on $I$, i.e., $\mathrm{Ker}(\Psi^T_J) \bigcap \mathrm{Ker}(\Phi)=\{0\}$, and there exist $y\in \RR^l$ and $\beta \in\RR^m$ such that $\Psi y = \Phi ^T\beta$, and
\begin{equation}\label{addcond2}
y_I=\mathrm{sign}(\Psi_I^Tx^*), ~~ \|y_J\|_\infty< 1.
\end{equation}
Then there exist $\rho$ and $\tau$ such that
\begin{align}\label{addcond1}
\|\Psi^T_I x\|_2\leq \rho\|\Psi^T_Jx\|_1+\tau \|\Phi x\|_2,
\end{align}
for all $x\in\RR^n$, and every minimizer
$x_{\delta}$ of $\|\Psi^Tx\|_1$ subjective to $\|\Phi x-b\|_2\leq \delta$ satisfies
\begin{equation}\label{addresult}
\|\Psi^T(x_{\delta}- x^*)\|_2\leq\frac{2(1+\rho)}{1-\|y_J\|_\infty}\|\Psi^T_Jx^*\|_1+\left(\frac{2(1+\rho)\|\beta\|_2}{1-\|y_J\|_\infty}+2\tau\right)\delta.
\end{equation}
\end{theorem}

\begin{remark}With Assumption \ref{assumption3}, part (1) of Condition~\ref{cond1} is equivalent to the existence of $\rho$ and $\tau$ in~\eqref{addcond1} for all $x\in\RR^n$ with $\rho$ and $\tau$. When part (1) of Condition~\ref{cond1} is satisfied, the existence of $\rho$ and $\tau$ is easy to show. If part (1) is not satisfied, then there exist $x\neq 0$ such that $\Psi^T_Jx=0$ and $\Phi x=0$. Assumption~\ref{assumption3} gives us that $\Psi^T_Ix\neq 0$. Thus $\|\Psi_I^Tx\|_2>0=\rho \|\Psi_J^Tx\|_1+\tau \|\Phi x\|_2$ for any $\rho$ and $\tau$.

Condition~\eqref{addcond2} can be relaxed into
\begin{equation*}
\|y_I-\mathrm{sign}(\Psi_I^Tx^*)\|_2\leq \theta_1, ~~ \|y_J\|_\infty< 1.
\end{equation*}
In this case, denote $\mu_1:=\rho\theta_1+\|y_J\|_\infty, \mu_2:=\tau\theta_1+\|\beta\|_2$. If $\mu_1<1$, then every minimizer
$x_{\delta}$ of $\|\Psi^Tx\|_1$ subjective to $\|\Phi x-b\|_2\leq \delta$ satisfies
\begin{equation*}
\|\Psi^T(x_{\delta}- x^*)\|_2\leq\frac{2(1+\rho)}{1-\mu_1}\|\Psi^T_Jx^*\|_1+\left(\frac{2(1+\rho)\mu_2}{1-\mu_1}+2\tau\right)\delta.
\end{equation*}
\end{remark}

%\begin{remark}
%Because the factor $1/\sqrt{s}$ is saved, the assumed condition \eqref{addcond1} is similar to but strictly weaker than the $\ell_1$-stable $\Psi^T$-null space property
%$$ \|\Psi_S^Tw\|_2\leq \frac{\rho}{\sqrt{s}} \|\Psi^T_P w\|_1+\tau\|\Phi w\|_2, \forall w\in \RR^m$$
%which was defined in \cite{rauhut2013} as a generation of the $\ell_1$-stable null space property in \cite{foucart2013,foucart2014}. Moreover,  it is not hard to see that condition \eqref{addcond1} implies $\mathrm{Ker}(\Psi^T_P) \bigcap \mathrm{Ker}(\Phi)=\{0\}$ under Assumption \ref{assumption3}. So it is stronger than part (1) of Condition \ref{cond1}. However, compared to part (2) of Condition \ref{cond1}, the restriction on $y$ in Theorem \ref{thm3} has been relaxed as we introduce the parameters $\theta_1, \theta_2$.
%\end{remark}

\begin{remark}
Theorem \ref{thm3} is inspired by several existing results in the standard $\ell_1$-synthesis case; for example Theorem 3.1 in \cite{rauhut2014} and Theorem 4.33 and Exercise 4.17 in \cite{foucart2013}. The result here can be viewed as an extension from the  $\ell_1$-synthesis   case to the  $\ell_1$-analysis  case.
\end{remark}

\section{Related Works}
\label{sec:3}
In the case of $\Psi=Id$, Condition~\ref{cond1} is well known in the literature for $\ell_1$ (or $\ell_1$-synthesis) minimization. It is initially proposed in \cite{Fuch1} as a sufficient condition for the $\ell_1$ solution uniqueness. For problems (\ref{eq:qpa}) and (\ref{eq:qpb}), \cite{Fuch2,Tibs} present sufficient but non-necessary conditions for solutions uniqueness. Later, its necessity is established in \cite{Gras1} for model \eqref{eq:qpa} and then in \cite{ZYC} for all models in \eqref{mod1}, assuming $\Psi=Id$ or equal to an orthogonal basis. The solution robustness of model \eqref{eq:qpa} is given under the same condition in \cite{Gras1}. Below we restrict our literature review to results for the $\ell_1$-analysis model.

\subsection{Previous Uniqueness Conditions}
Papers \cite{Gras2,NDEG,VPDF,Halt,Fadili_Peyre_Vaiter_Deledalle_Salmon13} establish the uniqueness of the $\ell_1$-analysis model and some use stronger conditions than ours. Our purpose in the comparison is restricted to uniqueness and robustness to arbitrary noise. The reader should be aware that some of them also imply stronger results such as sign consistency by Conditions~\ref{cond4} and~\ref{cond5} below; while the rest such as Conditions~\ref{cond2} and~\ref{cond3} below are not  used to derive stronger results yet.

The following condition in \cite{NDEG} guarantees the solution uniqueness for problem (\ref{eq:bp}):
\begin{condition}\label{cond2}
Given $\bar{x}$, let $Q$ be a basis matrix of $\mathrm{Ker}(\Phi)$, and $I=\mathrm{supp}(\Psi^T\bar{x})$. The followings are met:

(1) $\Psi^T_{I^c}Q$ is full column rank;

(2)$\|(Q^T\Psi_{I^c})^+Q^T\Psi_I\mathrm{sign}(\Psi^T_I\bar{x})\|_\infty<1$.
\end{condition}
Here $A^+$ is the pseudo-inverse of matrix A defined as $A^+=A^T(AA^T)^{-1}$.

Paper \cite{VPDF} proposes the following condition for the solution uniqueness and robustness for problems (\ref{eq:bp}) and (\ref{eq:qpa}) (the robustness requires the non-zero entries of $\Psi^T_I \title{x}$ to be sufficiently large compared to noise).
\begin{condition}\label{cond4}
For a given $\bar{x}$, index sets $I=\mathrm{supp}(\Psi^T\bar{x})$ and $J=I^c$ satisfy:

(1)$\mathrm{Ker}(\Psi^T_J) \bigcap \mathrm{Ker}(\Phi)=\{0\}$;

(2)Let $A^{[J]}=U(U^T\Phi^T\Phi U)^{-1}U^T$ and $\Omega^{[J]}=\Psi^+_J(\Phi^T\Phi A^{[J]}-Id)\Psi_I$, where $U$ is a basis matrix of $\mathrm{Ker}(\Psi^T_J)$. Then $$IC(\mathrm{sign}(\Psi^T_I\bar{x})):=\min_{u\in\mathrm{Ker}(\Psi_J)}\|\Omega^{[J]}\mathrm{sign}(\Psi^T_I\bar{x})-u\|_\infty<1.$$
\end{condition}

According to \cite{VPDF}, Conditions~\ref{cond2} and \ref{cond4} do not contain each other.

The following example shows that Conditions~\ref{cond2} and~\ref{cond4} are both stronger than Conditions~\ref{cond1} and~\ref{cond1'}.
Let
\begin{align*}
\Psi =\left(\begin{array}{ccc}10.5 & 1 & 10 \\0 & 1 & 0 \\ 0 & 0 & 1\end{array}\right), \quad\Phi =\left(\begin{array}{ccc} 0 & 1 & 0 \\ 0 & 0 & 1\end{array}\right),\quad \hat{x}=\left(\begin{array}{c}1\\-1 \\ -10\end{array}\right),\quad b=\left(\begin{array}{c}-1\\-10\end{array}\right).
\end{align*}
It is straightforward to verify that Conditions~\ref{cond1} and~\ref{cond1'} hold. However, Conditions~\ref{cond2} and~\ref{cond4} fail to hold.
Indeed, we have $\Psi^T\tilde{x}=(10.5, 0, 0)^T$ and $I=\{1\}$. $Q=(1,0,0)^T$ is a basis matrix of $\mathrm{Ker}(\Phi)$. Thus \begin{align*}\|(Q^T\Psi_{I^c})^+Q^T\Psi_I\mathrm{sign}(\Psi^T_I\tilde{x})\|_{\infty}=\left\|\left({10.5\over101}, {105\over101}\right)^T\right\|_\infty={105\over 101}.\end{align*}
Hence, Condition~\ref{cond2} does not hold. Furthermore, $U=(1,-1,-10)^T$ is a basis matrix of $\mathrm{Ker}(\Psi^T_J)$, and the definition of $\Omega^{[J]}$ gives us $\Omega^{[J]}=({10.5\over101}, {105\over101})^T$. Therefore, $IC(\mathrm{sign}(\Psi^T_I\tilde{x}))={105\over 101}>1$, so Condition~\ref{cond4} does not hold either. Paper \cite{VPDF} also presents sufficient conditions for solution uniqueness, which are reviewed in \cite{ZYC} and shown to be not necessary. %In \cite{VPDF}, Lemma 3 includes the following condition.

During the time we are preparing this manuscript, work~\cite{Fadili_Peyre_Vaiter_Deledalle_Salmon13} gives the same condition as Condition~\ref{cond1} and shows its sufficiency. However. The necessity of Condition~\ref{cond1} or~\ref{cond1'} is never discussed in the literature.

\subsection{Previous Robustness Conditions}
Turning to solution robustness, \cite{Gras2,Halt,Fadili_Peyre_Vaiter_Deledalle_Salmon13} have studied the robustness of problems (\ref{eq:qpa}) and (\ref{eq:qpb}) in the Hilbert-space setting. Translating to the finite dimension, the condition in \cite{Gras2} is equivalent to Condition \ref{cond1'}. %In their language, part (2) of Condition \ref{cond1'} can be restated as:
%\textit{There exists $y\in \partial \|\cdot\|_1(\Psi^T\bar{x})$ such that $\Psi y\in \mathrm{Im}(\Phi^T)$ and $J=\{i: |y_i|<1\}\neq \emptyset$.}
%\begin{condition}\label{cond3'}
%Given $\bar{x}$, the following two statements hold:
%
%(1) $\mathrm{Ker}(\Psi^T_J) \bigcap \mathrm{Ker}(\Phi)=\{0\}$;
%
%(2) There exists $y\in \partial \|\cdot\|_1(\Psi^T\bar{x})$ such that $\Psi y\in \mathrm{Im}(\Phi^T)$ and $J=\{i: |y_i|<1\}\neq \emptyset$.
%\end{condition}
Under Condition~\ref{cond1'}, work \cite{Gras2} shows the existence of constant $C$ (not explicitly given) such that the solution $x_{\delta,\lambda}$ to (\ref{eq:qpa}) obeys $\|\Psi^T(x_{\delta,\lambda}- x^*)\|_2\leq C\delta$ when $\lambda$ is set proportional to the noise level $\delta$. \cite{Fadili_Peyre_Vaiter_Deledalle_Salmon13} gives an explicit formula of $C$ in $\|x_{\delta,\lambda}- x^*\|_2\leq C\delta$ for solution $x_{\delta,\lambda}$ to (\ref{eq:qpa}).
%The corresponding robustness result is:
%\begin{theorem}[]\label{thm4'} (\cite{Gras2}, Theorem 4.4)
%Let $x^*$ be a fixed vector with $I=\mathrm{supp}(\Psi^Tx^*)$ and $b=\Phi x^*+w$ with $\|w\|_2\leq \delta$. Suppose that Condition \ref{cond3'} holds for $x^*$. Then, for $\lambda$ proportional to $\delta$, there exists $C$ independent of $\delta$ such that every minimizer $x_{\delta,\lambda}$ of (\ref{eq:qpa}) satisfies $\|\Psi^T(x_{\delta,\lambda}- x^*)\|_2\leq C\delta$.
%\end{theorem}
In order to obtain an explicit formula for $C$, \cite{Halt} introduces the following:
\begin{condition}\label{cond3}
Assume that $\Psi\Psi^T$ is invertible and let $\hat{\Psi}=(\Psi\Psi^T)^{-1}\Psi$. Given $\bar{x}$, the following two statements  hold:

(1) There exists some $y\in \partial \|\cdot\|_1(\Psi^T\bar{x})$ such that $\Psi y\in \mathrm{Im}(\Phi^T)$;

(2) For some $t\in (0, 1)$, letting $I(t)=\{i: |y_i|>t\}$, the mapping  $\hat{\Phi}:=\Phi|_{\mathrm{Span}\{\hat{\Psi}_i: i\in I(t)\}}$ is injective.
\end{condition}
Under this condition, the solutions to \eqref{eq:qpa} and \eqref{eq:qpb} are subject to error bounds whose constants depend on $t$, $\hat{\Phi}$, and other quantities.
\begin{proposition}\label{prop2}
Condition~\ref{cond3} is stronger than Condition~\ref{cond1'}.
\end{proposition}
\begin{proof}
Let $J=I(t)^c$; then we have $\|y_J\|_\infty\leq t<1$ from the definition of $I(t)$. It remains to show that $\textrm{Ker}(\Psi_J^T)\bigcap\textrm{Ker}(\Phi)=\{0\}$. For any $x\in\textrm{Ker}(\Psi_J^T)$, we have
\begin{align*}
x=(\Psi\Psi^T)^{-1}\Psi\Psi^Tx = (\Psi\Psi^T)^{-1}\Psi_J\Psi_J^Tx+(\Psi\Psi^T)^{-1}\Psi_{J^c}\Psi^T_{J^c}x=(\Psi\Psi^T)^{-1}\Psi_{I(t)}\Psi^T_{I(t)}x.
\end{align*}
Since $\Phi$ restricted to $\textrm{Span}\{\hat{\Psi}_i:i\in I(t)\}=\textrm{Im}((\Psi\Psi^T)^{-1}\Psi_{I(t)})$ is injective, we have what we need.
\qed\end{proof}

Definition 5 in paper \cite{VPDF} provides a much stronger condition below that strengthens Condition~\ref{cond4} by dropping the dependence on the $\Psi$-support (see the definition of $RC(I)$ below). %It strengthens Condition~\ref{cond4} and generalizes the ERC of Tropp \cite{Tropp2}. The condition and robustness result are as follows:
\begin{condition}\label{cond5}
Given $\bar{x}$, index sets $I=\mathrm{supp}(\Psi^T\bar{x})$ and $J=I^c$ satisfy:

(1) $\mathrm{Ker}(\Psi^T_J) \bigcap \mathrm{Ker}(\Phi)=\{0\}$;

(2) Letting $\Omega^{[J]}$ be given as in Condition \ref{cond4}, $$RC(I):=\max_{p\in \RR^{|I|}, \|p\|_\infty\leq 1}\min_{u\in\textrm{Ker}(\Psi_J)}\|\Omega^{[J]}p-u\|_\infty<1.$$
\end{condition}
Under this condition, a nice error bound and a certain kind of ``weak'' sign consistency (between $\Psi^Tx_{\delta,\lambda}$ and $\Psi^T x^*$) are given provided that problem \eqref{eq:qpa} is solved with the parameter $\lambda=\frac{\rho \|w\|_2 c_J}{2(1-RC(I))}$ for some $\rho>1$, where $c_J=\|\Psi^+_J\Phi^T(\Phi A^{[J]}\Phi^T-Id)\|_{2,\infty}$. When $1-RC(I)$ gets close to 0, this $\lambda$ can become too large than it should be.

\section{Verifying the Conditions}
\label{sec:4}
%Although there are many conditions addressing the uniqueness and robustness of problems (\ref{eq:bp})-(\ref{eq:qpb}), few methods exist to verify these conditions.
In this section, we present a method to verify Condition \ref{cond1}. Our method includes two steps:

(\textbf{Step 1:}) Let $\Phi=U\Sigma V^T$ be the singular value decomposition of $\Phi$. Assume $V=[v_1,\cdots,v_n]$. Since $\Phi$ has full row-rank, we have $\textrm{Ker}(\Phi)=\textrm{Span}\{v_{m+1}, \cdots, v_n\}$ and  $Q=[v_{m+1}, \cdots, v_n]$ as a basis  of $\textrm{Ker}(\Phi)$. We verify that $\Psi_J^TQ$ has full row-rank, ensuring part (1) of Condition \ref{cond1}.

(\textbf{Step 2:}) Let $u_1=-Q^T\Psi_I\textrm{sign}(\Psi^T_I\bar{x})$ and $A=Q^T\Psi_J$. Solve the convex problem
\begin{equation}\label{prob01}
\Min_{u\in\RR^{|J|}}\|u\|_\infty, \quad\st~~Au=u_1.
\end{equation}
If the optimal objective of (\ref{prob01}) is strictly less than 1, then part (2) of Condition \ref{cond1} holds. In fact, we have:
\begin{proposition}\label{prop5}
Part (2) of Condition \ref{cond1} holds if and only if (\ref{prob01}) has an  optimal objective $<1$.
\end{proposition}
\begin{proof}
Let $\hat{u}$ be a minimizer of (\ref{prob01}). Assume $\|\hat{u}\|_\infty<1$. We consider the vector $y$ composed by $y_I=\textrm{sign}(\Psi^T_I\bar{x})$ and $y_J=\hat{u}$. To show part (2) of Condition \ref{cond1}, it suffices to prove $\Psi y\in \textrm{Im}(\Phi^T)$, or equivalently, $Q^T\Psi y=0$. Indeed,

\begin{equation}
Q^T\Psi y = Q^T\Psi_Jy_J+Q^T\Psi_Iy_I =Q^T\Psi_J\hat{u}+Q^T\Psi_Iy_I=0\nonumber.
\end{equation}
The converse is obvious.
\qed\end{proof}
Convex program \eqref{prob01} is similar in form to one in \cite[Definition 4]{VPDF} though they are used to verify different conditions.

% There are somethings in common between our method described above and that in work \cite{VPDF}; note that the first item of Condition \ref{cond1} was not checked there. The authors in \cite{VPDF} reduced the checking of $IC(\textrm{sign}(\Psi^T_I\bar{x}))<1$ into the following convex program
% \begin{equation}\label{prob02}
% \Min_{u\in\RR^{|J|}}\|\Omega^{[J]}\textrm{sign}(\Psi^T_I\bar{x})-u\|_\infty \quad\st~~\Psi_Ju=0.
% \end{equation}
% If the objective above is strictly less than one, then $IC(\textrm{sign}(\Psi^T_I\bar{x}))<1$ from the definition of $IC(\cdot)$. To compare the convex problems (\ref{prob01}) and (\ref{prob02}), we write (\ref{prob02}) into the following equivalent form:
% \begin{equation}\label{prob03}
% \Min_{u\in\RR^{|J|}}\|u\|_\infty \quad\st~~\Psi_Ju=u_2,
% \end{equation}
% where $u_2=\Psi_J\Omega^{[J]}\textrm{sign}(\Psi^T_I\bar{x})$. The expression of $\Omega^{[J]}$ is very complicated. To get $\Omega^{[J]}$, two operators need to be dealt with. The first is to compute $\Psi_J^+$ which can be done by using singular value decomposition of $\Psi_J$. The second is to compute $A^{[J]}$ which can done by solving an optimization problem
% \begin{equation}\label{prob04}
% A^{[J]}u=\arg\min_{\Psi^T_Jx=0}\frac{1}{2}\|\Phi x\|_2^2-\langle x, u\rangle.
% \end{equation}
% Therefore, (\ref{prob01}) is much easier to solve than (\ref{prob03}) since the constant matrix $A$ and vector $u_1$ in the former can be obtained with only a singular value decomposition of $\Phi$.

\section{Proof of Theorem 1}
\label{sec:5}
We establish Theorem \ref{thm1} in two steps. Our first step proves the theorem for problem (\ref{eq:bp}) only. The second step proves Theorem \ref{thm1} for problems (\ref{eq:qpa}) and (\ref{eq:qpb}).

\subsection{Proof of Theorem \ref{thm1} for problem (\ref{eq:bp})}
The equivalence of the three statements is shown in the following orders: $3)\Longrightarrow 1) \Longrightarrow 2) \Longrightarrow 3)$.

$3)\Longrightarrow1)$. Consider any perturbation $\hat{x}+h$ where $h\in \textrm{Ker}(\Phi)\backslash\{0\}$. Take a subgradient $g\in \partial \|\cdot\|_1(\Psi^T\hat{x})$ obeying $g_I=\textrm{sign}(\Psi_I^T\hat{x})=y_I$, $g_K =y_K$, and $\|g_J\|_\infty\leq 1$ such that $\langle g_J, \Psi^T_J h\rangle =\|\Psi^T_Jh\|_1$. Then,
\begin{subequations}
\begin{align}
\|\Psi^T(\hat{x}+h)\|_1 &\geq  \|\Psi^T\hat{x}\|_1  + \langle \Psi g, h\rangle \\
&= \|\Psi^T\hat{x}\|_1  + \langle \Psi g-\Psi y, h\rangle \label{eq1}\\
&= \|\Psi^T\hat{x}\|_1  + \langle g-y, \Psi^Th\rangle\\
&= \|\Psi^T\hat{x}\|_1  +\langle g_J-y_J, \Psi_J^Th\rangle \label{eq2}\\
&\geq \|\Psi^T\hat{x}\|_1  + \|\Psi_J^Th\|_1(1-\|y_J\|_\infty)\label{eq3},
\end{align}
\end{subequations}
where (\ref{eq1}) follows from $\Psi y\in \textrm{Im}(\Phi^T)=\textrm{Ker}(\Phi)^\bot$ and $h\in \textrm{Ker}(\Phi)$, (\ref{eq2}) follows from the setting of $g$, and (\ref{eq3}) is an application of the inequality $\langle x, y\rangle\leq \|x\|_1\|y\|_\infty$ and $\langle g_J,\Psi^T_Jh\rangle=\|\Psi^T_Jh\|_1$. Since $h\in \textrm{Ker}(\Phi)\backslash\{0\}$ and $\textrm{Ker}(\Psi^T_J) \bigcap \textrm{Ker}(\Phi)=\{0\}$, we have $\|\Psi_J^Th\|_1>0$. Together with the condition $\|y_J\|_\infty<1$, we have $\|\Psi^T(\hat{x}+h)\|_1>\|\Psi^T\hat{x}\|_1$ for every $h\in \textrm{Ker}(\Phi)\backslash\{0\}$ which implies that $\hat{x}$ is the unique minimizer of (\ref{eq:bp}).

$1)\Longrightarrow2)$. For every $h\in \textrm{Ker}(\Phi)\backslash\{0\}$, we have $\Phi(\hat{x}+th)=\Phi\hat{x}$ and can find $t$ small enough around 0 such that $\textrm{sign}(\Psi_I^T(\hat{x}+th))=\textrm{sign}(\Psi^T_I\hat{x})$. Since $\hat{x}$ is the unique solution, for small and nonzero $t$ we have
\begin{subequations}
\begin{align*}
\|\Psi^T(\hat{x})\|_1 &< \|\Psi^T(\hat{x}+th)\|_1 =\|\Psi_I^T(\hat{x}+th)\|_1+\|\Psi_{I^c}^T(\hat{x}+th)\|_1\\
&= \langle \Psi_I^T(\hat{x}+th), \textrm{sign}(\Psi_I^T(\hat{x}+th))\rangle+\|t\Psi_{I^c}^Th\|_1\\
&= \langle \Psi_I^T\hat{x}+t\Psi_I^Th, \textrm{sign}(\Psi_I^T\hat{x})\rangle+\|t\Psi_{I^c}^Th\|_1\\
&= \langle \Psi_I^T\hat{x}, \textrm{sign}(\Psi_I^T\hat{x})\rangle+t\langle \Psi_I^Th, \textrm{sign}(\Psi_I^T\hat{x})\rangle+\|t\Psi_{I^c}^Th\|_1\\
&= \|\Psi^T(\hat{x})\|_1+t\langle \Psi_I^Th, \textrm{sign}(\Psi_I^T\hat{x})\rangle+\|t\Psi_{I^c}^Th\|_1.
\end{align*}
\end{subequations}
Therefore, for any $h\in \textrm{Ker}(\Phi)\backslash\{0\}$, we have
\begin{equation}\label{null}
\langle \Psi_I^Th, \textrm{sign}(\Psi_I^T\hat{x})\rangle<\|\Psi_{I^c}^Th\|_1.
\end{equation}
If the condition $\textrm{Ker}(\Psi^T_{I^c}) \bigcap \textrm{Ker}(\Phi)=\{0\}$ does not hold, we can choose a nonzero vector $h\in\textrm{Ker}(\Psi^T_{I^c})\bigcap\textrm{Ker}(\Phi)$. We also have $-h\in\textrm{Ker}(\Psi^T_{I^c})\bigcap\textrm{Ker}(\Phi)$. Then we have $\langle \Psi_I^Th, \textrm{sign}(\Psi_I^T\hat{x})\rangle<0$ and $-\langle \Psi_I^Th, \textrm{sign}(\Psi_I^T\hat{x})\rangle<0$, which is a contradiction.

It remains to show the existence of $y$ in item (2) of Condition \ref{cond1}. This part is in spirit of the methods in papers \cite{Gras1} and \cite{ZYC}, which are based on linear programming strong duality. We take $\hat{y}$ with restrictions $\hat{y}_I=\textrm{sign}(\Psi_I^T\hat{x})$ and $\hat{y}_{I^c}=0$. If such $\hat{y}$ satisfies $\Psi\hat{y}\in \textrm{Im}(\Phi^T)$, then the existence has been shown. If $\Psi\hat{y}\notin \textrm{Im}(\Phi^T)=\textrm{Ker}(\Phi)^\bot$, then we shall construct a new vector to satisfy part (2) of Condition \ref{cond1}. Let $Q$ be a basis matrix of $\textrm{Ker}(\Phi)$. We have that $a:=Q^T\Psi\hat{y}$ must be a nonzero vector. Consider the following problem
\begin{equation}\label{prob1}
\Min_{z\in \RR^l}\|z\|_\infty \quad\st~Q^T\Psi z=-a ~~\text{and}~~z_I=0.
\end{equation}
For any minimizer $\hat{z}$ of problem (\ref{prob1}), we have $\Psi(\hat{y}+\hat{z})\in \textrm{Ker}(\Phi)^\bot=\textrm{Im}(\Phi^T)$ and $(\hat{y}+\hat{z})_I=\hat{y}_I=\textrm{sign}(\Psi_I^T\hat{x})$. Thus, we shall show that the objective of problem (\ref{prob1}) is strictly less than 1. To this end, we rewrite  problem (\ref{prob1}) in an equivalent form as:
\begin{equation}\label{prob2}
\Min_{z}\|z_{I^c}\|_\infty \quad\st~Q^T\Psi_{I^c}z_{I^c}=-a,
\end{equation}
whose Lagrange dual problem is
\begin{equation}\label{prob3}
\Max_p \langle p, a\rangle \quad\st~\|\Psi_{I^c}^TQp\|_1\leq 1.
\end{equation}
Note that $Qp\in \textrm{Ker}(\Phi)$ and $|\langle p, a\rangle|=|\langle p, Q^T\Psi\hat{y}\rangle|=|\langle p, Q^T\Psi_I\textrm{sign}(\Psi^T_I\hat{x})\rangle|=|\langle \Psi^T_I Qp, \textrm{sign}(\Psi^T_I\hat{x})\rangle|$. By using (\ref{null}), for any $p$ we have
\begin{equation*}
|\langle p, a\rangle|=\left\{\begin{array}{ll}
|\langle \Psi^T_I Qp, \textrm{sign}(\Psi^T_I\hat{x})\rangle|=0,& \textrm{if}~~Qp=0;\\
|\langle \Psi^T_I Qp, \textrm{sign}(\Psi^T_I\hat{x})\rangle|<\|\Psi_{I^c}^TQp\|_1\leq 1, & \textrm{otherwise}.
\end{array}\right.
\end{equation*}
Hence, problem \eqref{prob3} is feasible, and its objective value is strictly less than 1. By the linear programming strong duality property,  problems \eqref{prob1} and \eqref{prob2} also have solutions, and their the  objective value is strictly less than 1, too.  This completes the proof.

$2)\Longrightarrow3)$. Let $J=I^c$ and $K =\emptyset$; then Condition~\ref{cond1'} follows.

The proof of $3)\Longrightarrow1)$ is a standard technique in compressed sensing community.

\subsection{Proof of Theorem \ref{thm1} for problems (\ref{eq:qpa}) and (\ref{eq:qpb})}
\begin{lemma}\label{lem4}
Let $\gamma>0$. If $\gamma\|\Phi x-b\|_2^2+\|\Psi^Tx\|_1$ is constant on a convex set $\Omega$,  then both $\Phi x-b$ and $\|\Psi^Tx\|_1$ are constant on $\Omega$.
\end{lemma}

\begin{proof}
It suffices to prove the case where the convex set has more than one point. Suppose $x_1$ and $x_2$ are arbitrary two \emph{different} points in $\Omega$. Consider the line segment $L$ connecting $x_1$ and $x_2$. By the convexity of set $\Omega$, we know $L\subset \Omega$. Thus, $\hat{c}=\gamma\|\Phi x-b\|_2^2+\|\Psi^Tx\|_1$ is a constant on $L$. If $\Phi x_1-b\neq \Phi x_2-b$, then for any $0<\alpha <1$, we have
\begin{subequations}
\begin{align}
&\quad\gamma\|\Phi(\alpha x_1+(1-\alpha)x_2)-b\|_2^2+\|\Psi^T(\alpha x_1+(1-\alpha)x_2)\|_1\\
&= \gamma\|\alpha(\Phi x_1-b)+(1-\alpha)(\Phi x_2 -b)\|_2^2+\|\alpha(\Psi^T x_1)+(1-\alpha)(\Psi^T x_2)\|_1\\
&<\alpha (\gamma\|\Phi x_1-b\|_2^2 +\|\Psi^Tx_1\|_1) +(1-\alpha)(\gamma\|\Phi x_2-b\|_2^2+\|\Psi^Tx_2\|_1)\\
&=\alpha \hat{c}+(1-\alpha)\hat{c}=\hat{c},
\end{align}
\end{subequations}
where the strict inequality follows from the \emph{strict} convexity of $\gamma\|\cdot\|_2^2$ and the convexity of $\|\Psi^Tx\|_1$.  This means that the points $\alpha x_1+(1-\alpha)x_2$ on $L$ attain a lower value than $\hat{c}$, which is a contradiction. Therefore, we have $\Phi x_1-b= \Phi x_2-b$, from which it is easy to see $\|\Psi^Tx_1\|_1 = \|\Psi^Tx_2\|_1$.
\qed\end{proof}

We let $X_\lambda$ and $Y_\delta$ denote the sets of solutions to problems (\ref{eq:qpa}) and (\ref{eq:qpb}), respectively; moreover, we assume that these two sets are nonempty. Then, from Lemma \ref{lem4}, we have the following result.

\begin{corollary}\label{cor2}
In problem (\ref{eq:qpa}), $\Phi x-b$ and $\|\Psi^Tx\|_1$ are constant on $X_\lambda$; in problem (\ref{eq:qpb}), $\Phi x-b$ and $\|\Psi^Tx\|_1$ are constant on $Y_\delta$.
\end{corollary}

\begin{proof}
Since $\|\Phi x-b\|_2^2+\lambda\|\Psi^Tx\|_1$ is constant over $X_\lambda$, the result follows directly from Lemma~\ref{lem4} for problem~(\ref{eq:qpa}). For problem~(\ref{eq:qpb}), if $0\in Y_\delta$, then we have $Y_\delta=\{0\}$ because of the full row-rankness of $\Psi$. The result holds trivially. Suppose $0\not\in Y_\delta$. Since the optimal objective $\|\Psi^Tx\|_1$ is constant for all $x\in Y_\delta$, we have to show that $\|\Phi x-b\|_2^2=\delta$ for all $x\in Y_\delta$. If there exist a nonzero $\hat{x}\in Y_\delta$ such that $\|\Phi \hat{x}-b\|_2^2<\delta$, we can find a non-empty ball $\mathcal{B}$ centered at $\hat{x}$ with a sufficiently small radius $\rho>0$ such that $\|\Phi \tilde{x}-b\|_2^2<\delta$ for all $\tilde{x}\in\mathcal{B}$. Let $\alpha =\min\{{\rho\over 2\|\hat{x}\|_2},{1\over2}\}\in (0,1)$. We have $(1-\alpha)\hat{x}\in\mathcal{B}$ and $\|(1-\alpha)\Psi^T\hat{x}\|_1 <\|\Psi^T\hat{x}\|_1$, which is a contradiction.
\qed\end{proof}

\begin{proof}[Proof of Theorem \ref{thm1} for problems (\ref{eq:qpa}) and (\ref{eq:qpb})]
This proof exploits Corollary \ref{cor2}. Since the results of Corollary \ref{cor2} are identical for problems \eqref{eq:qpa} and \eqref{eq:qpb}, we present the proof for problem \eqref{eq:qpa} only.

By assumption, $X_\lambda$ is nonempty so we pick $\hat{x}\in X_\lambda$. Let $b^* = \Phi \hat{x}$, which is independent of the choice of $\hat{x}$ according to Corollary \ref{cor2}. We introduce the following problem
\beq\label{eq:dd}
\Min_x\|\Psi^Tx\|_1,\quad\st~\Phi x= b^*,
\eeq
and let $X^*$ denote its solution set.

Now, we show that $X_\lambda = X^*$. Since $\Phi x = \Phi \hat{x}$ and $\|\Psi^Tx\|_1 = \|\Psi^T\hat{x}\|_1$ for all $x\in X_\lambda$ and conversely any $x$ obeying $\Phi x=\Phi \hat{x}$ and $\|\Psi^T x\|_1=\|\Psi^T \hat{x}\|_1$  belongs to $X_\lambda$, it suffices to show that $\|\Psi^Tx\|_1 = \|\Psi^T\hat{x}\|_1$ for any ${x}\in X^*$. Assuming this does \emph{not} hold, then since problem \eqref{eq:dd} has  $\hat{x}$ as a feasible solution and has a finite objective,  we have a nonempty $X^*$ and there exists $\tilde{x}\in X^*$ satisfying $\|\Psi^T\tilde{x}\|_1 < \|\Psi^T\hat{x}\|_1$. But,  $\|\Phi \tilde{x}-b\|_2 = \|b^*-b\|_2=\|\Phi \hat{x}-b\|_2$ and $\|\Psi^T\tilde{x}\|_1 < \|\Psi^T\hat{x}\|_1$ mean that $\tilde{x}$ is a strictly better solution to problem \eqref{eq:qpa} than $\hat{x}$, contradicting the assumption $\hat{x}\in X_\lambda$.

Since $X_\lambda = X^*$,  $\hat{x}$ is the unique solution to problem \eqref{eq:qpa} if and only if it is the unique solution to problem \eqref{eq:dd}. Since problem \eqref{eq:dd} is in the same form of problem \eqref{eq:bp}, applying the part of Theorem \ref{thm1} for problem \eqref{eq:bp}, which is already proved, we conclude that  $\hat{x}$ is the unique solution to problem \eqref{eq:qpa} if and only if Condition \ref{cond1} or \ref{cond1'} holds.
\qed\end{proof}

\section{Proof of Theorem 2}
\label{sec:6}

\begin{lemma}\label{lem2}
Assume that vectors $\bar{x}$ and $y$ satisfy Condition \ref{cond1}. Let $I=\mathrm{supp}(\Psi^T\bar{x})$ and $J=I^c$. We have
\begin{equation}\label{ineq0}
\|\Psi^Tx-\Psi^T\bar{x}\|_1\leq C_3\|\Phi(x-\bar{x})\|_2+C_4 d_y(x,\bar{x}),\quad \forall x,
\end{equation}
where $d_y(x,\bar{x}):=\|\Psi^Tx\|_1-\|\Psi^T\bar{x}\|_1-\langle \Psi y, x-\bar{x}\rangle$ is the Bregman distance of function $\|\Psi^T\cdot\|_1$, the absolute constants $C_3, C_4$ are given in Theorem \ref{thm2}.
\end{lemma}

\begin{proof}
This proof is divided into two parts.
 They are partially inspired by \cite{Gras2}. %However, we derive explicit constants in terms of support index sets $I, J$, matrices $\Psi, \Phi$,  and vector $y$.
\begin{enumerate}
\item this part shows that for any $u\in \textrm{Ker}(\Psi^T_J)$,
\begin{equation}\label{ineq11}
\|\Psi^Tx-\Psi^T\bar{x}\|_1\leq  \left(1+\frac{C_3\|\Phi\|}{\sqrt{\lambda_{\min}(\Psi\Psi^T)}}\right)\|\Psi^T(x-u)\|_1 +C_3\|\Phi(x-\bar{x})\|_2.
\end{equation}
\item this part shows that
\begin{equation}\label{ineq22}
f(x):=\min \left\{ \|\Psi^T(x-u)\|_1 : ~u\in\textrm{Ker}(\Psi^T_J)\right\}\le(1-\|y_J\|_\infty)^{-1}d_y(x,\bar{x}).
\end{equation}
\end{enumerate}
Using the definition of $C_4$, combining \eqref{ineq11} and \eqref{ineq22} gives \eqref{ineq0}.

\emph{Part 1}. Let  $u\in \textrm{Ker}(\Psi^T_J)$. By the triangle inequality of norms, we get
\begin{equation}\label{eq0}
\|\Psi^Tx-\Psi^T\bar{x}\|_1\leq \|\Psi^T(x-u)\|_1 +\|\Psi^T(u-\bar{x})\|_1.
\end{equation}
Since $\bar{x}\in \textrm{Ker}(\Psi^T_J)$, we have $u-\bar{x}\in \textrm{Ker}(\Psi^T_J)$ and thus $\|u-\bar{x}\|_2\leq r(J)\|\Phi(u-\bar{x})\|_2$, where $r(J)<+\infty$ follows from part (1) of Condition \ref{cond1}. Using the fact that $\supp(\Psi^T(u- \bar{x}))=I$, we derive that
\begin{subequations}\label{eq11}
\begin{align}
\|\Psi^T(u- \bar{x})\|_1 &\leq  \sqrt{|I|}\|\Psi^T(u- \bar{x})\|_2\\
&\leq \sqrt{|I|}\|u-\bar{x}\|_2 \\
&\leq \sqrt{|I|}\, r(J)\|\Phi(u-\bar{x})\|_2\\
&= C_3\|\Phi(u-\bar{x})\|_2,
\end{align}
\end{subequations}
where we have used the assumption  $\lambda_{\max}(\Psi\Psi^T)=1$ and the definition $C_3=r(J)\sqrt{|I|}$. Furthermore,
\begin{subequations}\label{eq22}
\begin{align}
\|\Phi(u-\bar{x})\|_2&\leq  \|\Phi(x- u)\|_2+\|\Phi(x-\bar{x})\|_2\\
&\leq \|\Phi\|\|x- u\|_2+\|\Phi(x-\bar{x})\|_2\\
&\leq \frac{\|\Phi\|\|\Psi^T(x- u)\|_2}{\sqrt{\lambda_{\min}(\Psi\Psi^T)}}+\|\Phi(x-\bar{x})\|_2\\
&\leq \frac{\|\Phi\|\|\Psi^T(x- u)\|_1}{\sqrt{\lambda_{\min}(\Psi\Psi^T)}}+\|\Phi(x-\bar{x})\|_2.
\end{align}
\end{subequations}
Therefore, we get \eqref{ineq11} after combining \eqref{eq0}, \eqref{eq11}, and \eqref{eq22}.
% \begin{equation}
% \|\Psi^T(u- \bar{x})\|_1 \leq \frac{C_3\|\Phi\|}{\sqrt{\lambda_{\min}(\Psi\Psi^T)}}\|\Psi^T(x-u)\|_1 +C_3\|\Phi(x-\bar{x})\|_2.
% \end{equation}
% Together with (\ref{eq0}), we have
% \begin{equation}\label{ineq1}
% \|\Psi^Tx-\Psi^T\bar{x}\|_1\leq \min \left\{\left(1+\frac{C_3\|\Phi\|}{\sqrt{\lambda_{\min}(\Psi\Psi^T)}}\right)\|\Psi^T(x-u)\|_1 +C_3\|\Phi(x-\bar{x})\|_2: ~u\in\textrm{Ker}(\Psi^T_J)\right\}.
% \end{equation}

\emph{Part 2}.
% Consider the following convex program
%  \begin{equation}\label{opt4}
% f(x):=\min \{ \|\Psi^T(x-u)\|_1 : ~u\in\textrm{Ker}(\Psi^T_J)\}.
% \end{equation}
% In the following, we prove that
% \begin{equation}\label{eq4}
% f(x)\leq (1-\|y_J\|_\infty)^{-1}d_y(x,\bar{x});
% \end{equation}
% note that $\|y_J\|_\infty<1$ from part (2) of Condition \ref{cond1}.
Since $\langle \Psi y, \bar{x}\rangle =\|\Psi^T\bar{x}\|_1$ implies $d_y(x,\bar{x})=\|\Psi^Tx\|_1-\langle \Psi y, x \rangle$, it is equivalent to proving
 \begin{equation*}
f(x)\leq (1-\|y_J\|_\infty)^{-1}(\|\Psi^Tx\|_1-\langle \Psi y, x \rangle).
\end{equation*}
Since $u\in\textrm{Ker}(\Psi^T_J)$ is equivalent to $\Psi^T_Ju = 0$, the Lagrangian  of the minimization problem in \eqref{ineq22} is
\begin{equation*}L(u, v)=\|\Psi^T(x-u)\|_1 +\langle v, \Psi^T_J u\rangle =\|\Psi^T(x-u)\|_1  +\langle \Psi_J v,  u-x \rangle+ \langle \Psi_Jv, x\rangle.\end{equation*}
Then, $f(x)=\min_u\max_v L(u,v)$. Following the minimax theorem, we derive that
\begin{subequations}
\begin{align*}
&f(x)= \max_v\min_u L(u,v)\\
=&\max_v\min_u \{\|\Psi^T(x-u)\|_1  +\langle \Psi_J v,  u-x \rangle+ \langle \Psi_Jv, x\rangle\}\\
=&\max_w\min_u \{\|\Psi^T(x-u)\|_1  +\langle w,  u-x \rangle+ \langle w, x\rangle: w\in \textrm{Im}(\Psi_J)\}\\
=&\max_w \{ \langle w, x\rangle: w\in \partial \|\Psi^T\cdot\|_1(0)\cap\textrm{Im}(\Psi_J)\}\\
=&\max_w \{ \langle c\Psi y+ w, x\rangle: w\in \partial \|\Psi^T\cdot\|_1(0)\cap \textrm{Im}(\Psi_J)\}-\langle c\Psi y, x\rangle, ~~\forall c>0\\
=&c \max_w \{ \langle \Psi y+ w, x\rangle: w\in  c^{-1}\partial \|\Psi^T\cdot\|_1(0)\cap \textrm{Im}(\Psi_J)\}-c \langle \Psi y, x\rangle,  ~~\forall c>0.
\end{align*}
\end{subequations}
Let $$c=(1-\|y_J\|_\infty)^{-1}$$ and $Z_J=\{z\in \RR^l:z_I=0\}$.  Since $\|y_J\|_\infty<1$ from part (2) of Condition \ref{cond1}, we have $c<+\infty$ and get
\begin{equation*}
\left(y+c^{-1}\partial \|\cdot\|_1(0)\cap Z_J\right)\subset \partial \|\cdot\|_1(0),
\end{equation*}
from which we conclude
 \begin{equation*}
\left(\Psi y+c^{-1}\partial \|\Psi^T\cdot\|_1(0)\cap \textrm{Im}(\Psi_J)\right) \subset \partial \|\Psi^T\cdot\|_1(0).
\end{equation*}
Hence, for any $w\in c^{-1}\partial \|\Psi^T\cdot\|_1(0)\cap \textrm{Im}(\Psi_J)$, it holds $\Psi y+ w\subset \partial \|\Psi^T\cdot\|_1(0)$, which by the convexity of $\|\Psi^T\cdot\|_1$ implies $\langle \Psi y+ w, x\rangle\leq \|\Psi^Tx\|_1$.
Therefore, $f(x)\leq c(\|\Psi^Tx\|_1-\langle \Psi y, x\rangle)$.
\qed\end{proof}

\begin{lemma}\label{lem3}(\cite{BO}, Theorem 3; \cite{Gras1}, Lemma 3.5)
Suppose that $x^*\in \RR^n$ is a fixed vector obeying $\supp(\Psi^Tx^*)=I$ and that there are vectors satisfying
$y\in \partial \|\cdot\|_1(\Psi^Tx^*)$ and $\Psi y=\Phi^T \beta$. Then for every $\delta>0$ and every data vector $b$ satisfying $\|\Phi x^*-b\|_2\leq \delta$, the following two statements hold:

1)Every minimizer $x_{\delta,\lambda}$ of problem (\ref{eq:qpa}) satisfies $d_y(x_{\delta,\lambda}, x^*)\leq \frac{(\delta+\lambda\|\beta\|_2/2)^2}{\lambda}$
and $\|\Phi x_{\delta,\lambda}-b\|_2\leq \delta+\lambda\|\beta\|_2$;

2)Every minimizer $x_\delta$ of problem (\ref{eq:qpb}) satisfies $d_y(x_\delta, x^*)\leq 2\delta\|\beta\|_2$.
\end{lemma}

From $\Psi y=\Phi^T \beta$ and the full-rankness of $\Phi$, we have $\beta=(\Phi \Phi^T)^{-1}\Phi\Psi y$.

\begin{proof}[Proof of Theorem \ref{thm2}]
Firstly, we derive that
\begin{subequations}
\begin{align}
\|\Psi(x_{\delta,\lambda}-x^*)\|_1 &\leq  C_3\|\Phi(x_{\delta,\lambda}-x^*)\|_2+C_4 d_y(x_{\delta,\lambda},x^*) \\
&\leq C_3\|\Phi x_{\delta,\lambda}-b\|_2+C_3\|\Phi x^* -b\|_2+ C_4 d_y(x_{\delta,\lambda},x^*) \\
&\leq C_3( \delta+\lambda\|\beta\|_2) +C_3\delta +C_4\frac{(\delta+\lambda\|\beta\|_2/2)^2}{\lambda}\label{robust},
\end{align}
\end{subequations}
where the first and the third inequalities follow from Lemmas \ref{lem2} and \ref{lem3}, respectively. Substituting $\lambda=C_0\delta$ and collecting like terms in (\ref{robust}), we obtain the first part of Theorem \ref{thm2}. The second part can be proved in the same way.
\qed\end{proof}

\section{Proof of Theorem 3}\label{sec:8}

Let $w:=x_\delta-x^*, v:=\Psi^Tw$. Then by the minimality of $\|\Psi^Tx_\delta\|_1$, we derive
\begin{subequations}
\begin{align*}
\|\Psi^Tx^*\|_1 &\geq \|\Psi^Tx_\delta\|_1=\|v+\Psi^Tx^*\|_1=\|v_I+\Psi^T_Ix^*\|_1+ \|v_J+\Psi^T_Jx^*\|_1\\
 &\geq \langle v_I+\Psi^T_Ix^*, \sign(\Psi^T_Ix^*)\rangle+ \|v_J\|_1-\|\Psi^T_Jx^*\|_1\\
 &= \langle v_I, \sign(\Psi^T_Ix^*)\rangle + \|\Psi^T_Ix^*\|_1+\|v_J\|_1-\|\Psi^T_Jx^*\|_1.
\end{align*}
\end{subequations}
Rearranging the above inequality and using the fact that $\|\Psi^Tx^*\|_1= \|\Psi^T_Ix^*\|_1+\|\Psi^T_Jx^*\|_1$ yield
\begin{equation}\label{boundvp}
\|v_J\|_1\leq 2\|\Psi^T_Jx^*\|_1+|\langle v_I, \sign(\Psi^T_Ix^*)\rangle |.
\end{equation}
In addition, we have
\begin{align*}
|\langle v_I, \sign(\Psi^T_Ix^*)\rangle | =|\langle y_I, v_I\rangle|\leq |\langle y, v\rangle| + |\langle y_J, v_J\rangle| .
\end{align*}
In what follows, we further bound the right-hand side of the above inequality term by term. First of all, in view of the optimization constraint and the feasible of both $x_\delta$ and $x^*$, we have
$$ \|\Phi w\|_2=\|\Phi(x_\delta-x^*)\|_2\leq \|\Phi x_\delta-b\|_2+\|\Phi x^*-b\|_2\leq 2\delta.$$
%Thus, together with the inequality \eqref{addcond1}, we bound the first term as follows:
%\begin{equation}\label{boundvs}
%\|v_I\|_2=\|\Psi^T_Iw\|_2\leq  \rho \|\Psi^T_Jw\|_1 +\tau \|\Phi w\|_2\leq \rho \|v_J\|_1+2\tau\delta.
%\end{equation}
Using the condition $\Psi y=\Phi^T\beta$, we obtain
\begin{subequations}
\begin{align*}
|\langle y, v\rangle | & =|\langle  y, \Psi^Tw \rangle | =|\langle \Psi y,  w \rangle |=|\langle \Phi^T\beta, w \rangle | \\
 &= |\langle \beta, \Phi w \rangle |\leq \|\beta\|_2\cdot\|\Phi w\|_2\leq 2\delta\|\beta\|_2.
\end{align*}
\end{subequations}
At last, $|\langle y_J, v_J\rangle| \leq \|y_J\|_\infty\|v_J\|_1$. Collecting these upper bounds, we derive
$$|\langle v_I, \sign(\Psi^T_Ix^*)\rangle | \leq \|y_J\|_\infty\|v_J\|_1+2\|\beta\|_2\delta$$
Using \eqref{boundvp} and noting that $\|y_J\|_\infty<1$, we get
\begin{equation}\label{fboundvp}
\|v_J\|_1\leq \frac{2}{1-\|y_J\|_\infty}\|\Psi^T_Jx^*\|_1+ \frac{2\|\beta\|_2}{1-\|y_J\|_\infty}\delta.
\end{equation}
Now, using the inequality \eqref{addcond1}, we get
\begin{equation}\label{fboundvs}
\|v_I\|_2\leq \rho\|v_J\|_1+\tau\|\Phi w\|_2\leq  \frac{2\rho}{1-\|y_J\|_\infty}\|\Psi^T_Jx^*\|_1+ \left(\frac{2\rho\|\beta\|_2}{1-\|y_J\|_\infty}+2\tau\right)\delta.
\end{equation}
Finally, combing \eqref{fboundvp} and \eqref{fboundvs}, we obtain
\begin{subequations}
\begin{align*}
\|v\|_2 & \leq\|v_I\|_2+\|v_J\|_2\leq \|v_I\|_2+\|v_J\|_1 \\
 &\leq \frac{2(1+\rho)}{1-\|y_J\|_\infty}\|\Psi^T_Jx^*\|_1+\left(\frac{2(1+\rho)\|\beta\|_2}{1-\|y_J\|_\infty}+2\tau\right)\delta.
\end{align*}
\end{subequations}
This completes the proof.\qed

\section*{Acknowledgements}
The work of H. Zhang is supported by China Scholarship Council during his visit to Rice University, and in part by Graduate School of NUDT under Funding of Innovation B110202. The work of M. Yan is supported in part by the Center for Domain-Specific Computing (CDSC) funded by NSF grants CCF-0926127 and ARO/ARL MURI grant FA9550-10-1-0567. The work of W. Yin is supported in part by NSF grants DMS-0748839 and ECCS-1028790. They thank Rachel Ward and Xiaoya Zhang for their helpful discussions.

\bibliographystyle{unsrt}
\bibliography{analysisModel}

\end{document}